\documentclass{llncs}
\usepackage{graphicx,fancyhdr,amssymb,ifthen}
\usepackage[all]{xy}

\newcommand{\bR}{\mathbb{R}}

\newcommand{\bE}{\mathbb{E}}

\newcommand{\dom}{\mathrm{dom}\,}

\newcommand{\co}{\mathrm{co}\,}
\newcommand{\cP}{\mathcal{P}}
\newcommand{\Int}[1]{\mathrm{Int}({#1})}

\frontmatter

\begin{document}

\mainmatter              

\title{Asymmetric Topologies on Statistical Manifolds\thanks{This work was supported in part by BBSRC grant BB/L009579/1}}
\author{Roman V. Belavkin\inst{1}}
\authorrunning{Roman Belavkin}   
%
\tocauthor{Roman Belavkin}
\institute{School of Science and Technology\\
Middlesex University, London NW4 4BT, UK}

\maketitle              

\begin{abstract}
Asymmetric information distances are used to define asymmetric norms and quasimetrics on the statistical manifold and its dual space of random variables.  Quasimetric topology, generated by the Kullback-Leibler (KL) divergence, is considered as the main example, and some of its topological properties are investigated.
\end{abstract}


\section{Introduction}
It is difficult to overestimate the importance of the Kullback-Leibler (KL) divergence $D_{KL}[p,q]=\bE_p\{\ln(p/q)\}$ in probability and information theories, statistics and physics \cite{Kullback-Leibler51}.  Not only it plays the role of a non-symmetric squared Euclidean distance on the set $\cP(\Omega)$ of all probability measures on measurable set $(\Omega,\mathcal{A})$, satisfying the non-symmetric Pythagorean theorem \cite{Chentsov68} and the generalized law of cosines (see Theorem~\ref{th:cosines}), but it also possesses a number of other useful and often unique to it properties.  Indeed, it is G\^ateaux differentiable and strictly convex everywhere where it is finite (the convex cone of finite positive measures).  It is unique in the sense of additivity: $D_{KL}[p_1\otimes p_2,q_1\otimes q_2]=D_{KL}[p_1,q_1]+D_{KL}[p_2,q_2]$, and its Hessian defines Riemannian metric on the statistical manifold $\cP\subset Y_+$ invariant in the category of Markov morphisms.  The existence and uniqueness of this Riemannian metric is one of the most celebrated results in information geometry due to Chentsov (Lemma~11.3 in \cite{Chentsov72} or its infinite-dimensional version Theorem~5.1 in \cite{Morozova-Chentsov91}).

Perhaps, the only `inconvenient' property of the KL-divergence is its asymmetry: $D_{KL}[p,q]\neq D_{KL}[q,p]$ for some $p$ and $q$.  It means that a topology defined on $\cP(\Omega)$ in terms of the KL-divergence is not symmetric, and the analysis of asymmetric topological spaces (e.g. quasi-normed, quasi-metric or quasi-uniform spaces) is significantly more difficult than that of normed or metric spaces.  Many classical results about completeness, total boundedness or compactness do not hold in asymmetric topologies (e.g. see \cite{Fletcher-Lindgren82,Cobzas13}).  Perhaps, for this reason previous works have considered statistical manifolds as subsets of Banach spaces, such as the Orlicz spaces \cite{Pistone-Sempi95}.  This, of course, requires certain symmetrization.  Specifically, the Orlicz norm (or the equivalent Luxemburg norm) is defined using the integral of an even function $\phi(x)=\phi(-x)$ (called the $N$-function), which usually uses the absolute value $|x|=\max\{-x,x\}$ under the argument of $\phi$.  Because probability measures are positive functions, the transformation $x\mapsto|x|$ appears to be quite innocent and well-justified, because one can apply the highly developed theory of Banach spaces.  However, this may loose asymmetry that is quite natural in some random phenomena.  Moreover, symmetrization on the statistical manifold $\cP(\Omega)$ also automatically symmetrizes the topology of the dual space containing random variables.  When these random variables are used in the context of optimization (e.g. as utility or cost functions), their symmetrization is rather unnatural, and some random variables cannot be used.  Let us illustrate this in the following examples.

\begin{example}[The St. Petersburg lottery]
\label{ex:spb}
The lottery is played by tossing a coin repeatedly until the first head appears. The probability of head occurring on the $n$th toss assuming independent and identically distributed (i.i.d.) tosses of a fair coin is $q(n)=2^{-n}$.  If the payoff is $x(n)=2^n$, then the lottery has infinite expected payoff (this is historically the first example of unbounded expectation \cite{Bernoulli1738}).  If the coin is biased towards head, however, such as $p(n)=2^{-(1+\alpha)n}$, ($\alpha>0$), then the expected payoff becomes finite.  The effective domain of the moment generating function $\bE_q\{e^{\beta x}\}$ does not contain the ray $\{\beta x:\beta>0\}$, but it does contain the ray $\{-\beta x:\beta\geq0\}$.  Thus, random variable $x(n)=2^n$ belongs to the space, where zero is not in the interior of the effective domain of $\bE_q\{e^{\beta x}\}$ or of the cumulant generating function $\Psi_q(\beta x)=\ln\bE_q\{e^{\beta x}\}$ (the Legendre-Fenchel transform of $D_{KL}[p,q]$).  This implies that sublevel sets $\{p:D_{KL}[p,q]\leq\lambda\}$ in the dual space are unbounded (see Theorem~\ref{th:zero-interior-bounded} of \cite{Asplund-Rockafellar69,Moreau67}).  Note that exponential family distributions $p(x;\beta)=e^{\beta x-\Psi_q(\beta x)}q$ solve the problem of maximization of random variable $x$ on $\{p:D_{KL}[p,q]\leq\lambda\}$, while $p(-x;\beta)=e^{-\beta x-\Psi_q (-\beta x)}q$ solve minimization (i.e. maximization of $-x$).  This example illustrates asymmetry typical of optimization problems, because random variable $x(n)=2^n$ has bottom ($x(1)=2$), but it is topless.
\end{example}

\begin{example}[Error minimization]
\label{ex:error}
Consider the problem of minimization of the error function $z(a,b)$ (or equivalently maximization of utility $x=-z$), which can be defined using some metric $d:\Omega\times\Omega\to[0,\infty)$ on $\Omega$.  For example, using the Hamming metric $d_H(a,b)=\sum_{n=1}^l\delta_{b_n}(a_n)$ on finite space $\{1,\ldots,\alpha\}^l$ or using squared Euclidean metric $d_E^2(a,b)=\sum_{n=1}^l|a_n-b_n|^2$ on a real space $\bR^l$ (i.e. by defining utility $x=-d_H$ or $x=-\frac12d_E^2$).  Let $w$ be the joint distribution of $a$ and $b\in\Omega$, and let $q$, $p$ be its marginal distributions.  The KL-divergence $D_{KL}[w,q\otimes p]=:I_S(a,b)$ defines the amount of mutual information between $a$, $b$.  Joint distributions minimizing the expected error subject to constraint $I_S(a,b)\leq\lambda$ belong to the exponential family $w(x;\beta)=e^{\beta x-\Psi_{q\otimes p}(\beta x)}q\otimes p$.  With the maximum entropy $q\otimes p$ and Hamming metric $x=-d_H$, this $w(x;\beta)$ corresponds to the binomial distribution, and in the case of squared Euclidean metric $x=-\frac12d_E^2$ to the Gaussian distribution.  In the finite case of $\Omega=\{1,\ldots,\alpha\}^l$, the random variable $x=-d_H$ can be reflected $x\mapsto -x=d_H$, as both $\Psi_{q\otimes p}(-\beta d_H)$ and $\Psi_{q\otimes p}(\beta d_H)$ are finite (albeit possibly with different values).  However, in the infinite case of $\Omega=\bR^l$, the unbounded random variable $x=-\frac12d_E^2$ cannot be reflected, as maximization of Euclidean distance has no solution, and $\Psi_{q\otimes p}(\beta\frac12d_E^2)=\infty$ for any $\beta\geq0$.  As in the previous example, $0\notin\Int{\dom\Psi_{q\otimes p}}$.
\end{example}

The examples above illustrate that symmetrization of neighbourhoods on the statistical manifold requires random variables to be considered together with their reflections $x\mapsto -x$.  However, this is not always desirable or even possible in the infinite-dimensional case.  First, random variables used in optimization problems, such as utilities or cost functions, do not form a linear space, but a wedge.  Operations $x\mapsto -x$ and $x\mapsto|x|$ are not monotonic.  Second, the wedge of utilities or cost functions may include unbounded functions (e.g. concave utilities $x:\Omega\to\bR\cup\{-\infty\}$ and convex cost functions $z:\Omega\to\bR\cup\{\infty\}$).  In some cases, one of the functions $x$ or $-x$ cannot be absorbed into the effective domain of the cumulant generating function $\Psi$ (i.e. $\beta x\notin\dom\Psi$ for any $\beta>0$), in which case the symmetrization $x\mapsto|x|$ would leave such random variables out.

In the next section, we shall outline the main ideas for defining dual asymmetric topologies using polar sets and sublinear functions related to them.  In Section~\ref{sec:distance}, we shall introduce a generalization of Bregman divergence, a generalized law of cosines and define associated asymmetric seminorms and quasi-metrics.  In Section~\ref{sec:KL-topology}, we shall prove that asymmetric topology defined by the KL-divergence is complete, Hausdorff and contains a separable Orlicz subspace.

\section{Topologies Induced by Gauge and Support Functions}
\label{sec:support-functions}

Let $X$ and $Y$ be a pair of linear spaces over $\bR$ put in duality via a non-degenerate bilinear form $\langle\cdot,\cdot\rangle:X\times Y\to\bR$:
\[
\langle x,y\rangle=0\,,\ \forall\, x\in X\ \Rightarrow y=0\,,\qquad
\langle x,y\rangle=0\,,\ \forall\, y\in Y\ \Rightarrow x=0
\]
When $x$ is understood as a random variable and $y$ as a probability measure, then the pairing is just the expected value $\langle x,y\rangle=\bE_y\{x\}$.  We shall define topologies on $X$ and $Y$ that are compatible with respect to the pairing $\langle\cdot,\cdot\rangle$, but the bases of these topologies will be formed by systems of neighbourhoods of zero that are generally non-balanced sets (i.e. $y\in M$ does not imply $-y\in M$).  It is important to note that such spaces may fail to be topological vector spaces, because multiplication by scalar can be discontinuous (e.g. see \cite{Borodin01}).  Let us first recall some properties that depend only on the pairing $\langle\cdot,\cdot\rangle$.


Each non-zero $x\in X$ is in one-to-one correspondence with a hyperplane $\partial\Pi x:=\{y:\langle y,x\rangle=1\}$ or a closed halfspace $\Pi x:=\{y:\langle y,x\rangle\leq1\}$.  The intersection of all $\Pi x$ containing $M$ is the \emph{convex closure} of $M$ denoted by $\co[M]$.  Set $M$ is closed and convex iff $M=\co[M]$.  The \emph{polar} of $M\subseteq Y$ is
\[
M^\circ:=\{x\in X:\langle x,y\rangle\leq1\,,\ \forall\,y\in M\}
\]
The polar set is always closed and convex and $0\in M^\circ$.  Also, $M^{\circ\circ}=\co[M\cup\{0\}]$, and $M=M^{\circ\circ}$ if and only if $M$ is closed, convex and $0\in M$.  Without loss of generality we shall assume $0\in M$.  The mapping $M\mapsto M^\circ$ has the properties:
\begin{eqnarray}
(M\cup N)^\circ&=&M^\circ\cap N^\circ \label{eq:polar-or}\\
(M\cap N)^\circ&=&\co[M^\circ\cup N^\circ] \label{eq:polar-and}
\end{eqnarray}
We remind that set $M\subseteq Y$ is called:
\begin{list}{}{}
\item [\emph{Absorbing}] if $y/\alpha\in M$ for all $y\in Y$ and $\alpha\geq\varepsilon(y)$ for some $\varepsilon(y)>0$.
\item [\emph{Bounded}] if $M\subseteq\alpha\Pi x$ for any closed halfspace $\Pi x$ and some $\alpha>0$.
\item [\emph{Balanced}] if $M=-M$.
\end{list}
Set $M$ is absorbing if and only if its polar $M^\circ$ is bounded; If $M$ is balanced, then so is $M^\circ$.  If $M$ is closed and convex, then the following are balanced closed and convex sets: $-M\cap M$, $\co[-M\cup M]$.

Given set $M\subseteq Y$, $0\in M$, the set $Y_M:=\{y:y/\alpha\in M,\ \forall\,\alpha\geq\varepsilon(y)>0\}$ of elements absorbed into $M$ can be equipped with a topology, uniquely defined by the base of closed neighbourhoods of zero $\mathfrak{M}:=\{\alpha M:\alpha>0\}$.  The set $X_M:=\{x:\langle x,y\rangle\leq\alpha,\ \forall\,y\in M$\} of hyperplanes bounding $M$ are absorbed into the polar set $M^\circ$, and the collection $\mathfrak{M}^\circ:=\{\alpha^{-1}M^\circ:\alpha^{-1}>0\}$ is the base of the polar topology on $X_M$.  Note that $Y_M$ (resp. $X_M$) is a strict subset of $Y$ (resp. $X$), unless $M$ is absorbing (resp. bounded).  Moreover, it may fail to be a topological vector space, unless $M$ (or $M^\circ$) is balanced.  Such polar topologies can be defined using gauge or support functions.

The \emph{gauge} (or \emph{Minkowski} functional) of set $N\subseteq X$ is the mapping $\mu N:X\rightarrow\bR\cup\{\infty\}$ defined as
\[
\mu N(x):=\inf\{\alpha>0:x/\alpha\in N\}
\]
with $\mu N(0):=0$ and $\mu N(x):=\infty$ if $x/\alpha\notin N$ for all $\alpha>0$.  Note that $\mu N(x)=0$ if $x/\alpha\in N$ for all $\alpha>0$.  The following statements are implied by the definition.

\begin{lemma}
$\mu N(x)<\infty$ for all $x\in X$ if and only if $N$ is absorbing; $\mu N(x)>0$ for all $x\neq0$ if and only if $N$ is bounded.
\label{lm:guage}
\end{lemma}

The gauge is positively homogeneous function of the first degree, $\mu N(\beta x)=\beta\mu N(x)$, $\beta>0$, and if $N$ is convex, then it is also subadditive, $\mu N(x_1+x_2)\leq \mu N(x_1)+\mu N(x_2)$.  Thus, the gauge of an absorbing closed convex set satisfies all axioms of a seminorm apart from symmetry, and therefore it is a \emph{quasi-seminorm}.  Function $\rho_N(x_1,x_2)=\mu N(x_2-x_1)$ is a \emph{quasi-pseudometric} on $X$.  If $N$ is bounded, then $\mu N$ is a \emph{quasi-norm} and $d_N$ is a \emph{quasi-metric}.  Symmetry $\mu N(x)=\mu N(-x)$ and $\rho_N(x_1,x_2)= \rho_N(x_2,x_1)$ requires $N$ to be balanced.

The \emph{support} function of set $M\subseteq Y$ is the mapping $sM:X\to\bR\cup\{\infty\}$:
\[
sM(x):=\sup\{\langle x,y\rangle:y\in M\}
\]
Like the gauge, the support function is also positively homogeneous of the first degree, and it is always subadditive.  Generally, $\mu N(x)\geq sN^\circ(x)$, with equality if and only if $N$ is convex.  In fact, the following equality holds:
\begin{lemma}
$sM(x)=\mu M^\circ(x)$, $\forall\,M\subseteq Y,\ 0\in M$.
\label{lm:polar}
\end{lemma}
\begin{proof}
$\langle x,y\rangle\leq sM(x)$ for all $y\in M$, $sM(x/\alpha)=\alpha^{-1}sM(x)$, $sM(x):=\inf\{\alpha>0:\langle x/\alpha,y\rangle\leq1,\ \forall\,y\in M\}=\inf\{\alpha>0:x/\alpha\in M^\circ\}$.\qed
\end{proof}

The following is the asymmetric version of the H\"older inequality:
\begin{lemma}[Asymmetric H\"older]
$\langle x,y\rangle\leq sM(x)sM^\circ(y)$, $\forall\,M\subseteq Y,\ 0\in M$.
\label{lm:holder}
\end{lemma}
\begin{proof}
$\langle x,y\rangle\leq sM(x)$, $\langle x/sM(x),y\rangle\leq 1$ for all $y\in M$, so that $x/sM(x)\in M^\circ$.  Hence $\langle x/sM(x),y\rangle\leq sM^\circ(y)$.\qed
\end{proof}

The support function $sM(x)$ can be symmetrized in two ways:
\[
s^sM(x):=s[-M\cup M](x)\,,\qquad
s^\circ M(x):=s[-M\cap M](x)
\]

\begin{lemma}
\begin{enumerate}
\item $s^sM(x)\geq sM(x)\geq s^\circ M(x)$.
\item $s^sM(x)=\sup\{sM(-x),sM(x)\}$.
\item $s^\circ M(x)=\co[\inf\{sM(-x),sM(x)\}]=\inf\{sM(z)+sM(z-x):z\in X\}$.
\item $s^\circ M(x)=\sup\{\langle x,y\rangle:s^sM^\circ(y)\leq1\}$.
\end{enumerate}
\label{lm:even}
\end{lemma}

\begin{proof}
\begin{enumerate}
\item Follows from set inclusions: $-M\cup M\supseteq M\supseteq -M\cap M$.
\item $s^sM(x)=\mu[-M^\circ\cap M^\circ](x)=\sup\{\mu M^\circ(-x),\mu M^\circ(x)\}$ by Lemma~\ref{lm:polar} and equation~(\ref{eq:polar-or}).
\item $s^\circ M(x)=\mu\co[-M^\circ\cup M^\circ](x)=\co[\inf\{\mu M^\circ(-x),\mu M^\circ(x)\}]$ by Lemma~\ref{lm:polar} and equation~(\ref{eq:polar-and}).  The second equation follows from the equivalence of convex closure infimum and infimal convolution for sublinear functions \cite{Tikhomirov90:_convex}.
\item Follows from $sM(x)=sM^{\circ\circ}(x)$, $0\in M$, and $N^\circ=\{y:sN(y)\leq 1\}$ by substituting $N=(-M\cap M)^\circ=\co[-M^\circ\cup M^\circ]$.\qed
\end{enumerate}
\end{proof}

\section{Distance Functions and Sublevel Neighbourhoods}
\label{sec:distance}


A closed neighbourhood of $z\in Y$ can be defined by sublevel set $\{y:D[y,z]\leq\lambda\}$ of a \emph{distance} function $D:Y\times Y\to\bR\cup\{\infty\}$ satisfying the following axioms:
\begin{enumerate}
\item $D[y,z]\geq0$.
\item $D[y,z]=0$ if $y=z$.
\end{enumerate}
Thus, a distance is generally not a metric (i.e. non-degeneracy, symmetry or the triangle inequality are not required).  A distance function associated with closed functional $F:Y\to\bR\cup\{\infty\}$ can be defined as follows:
\begin{equation}
D_F[y,z]:=\inf\{F(y)-F(z)-\langle x,y-z\rangle:x\in\partial F(z)\}
\label{eq:distance}
\end{equation}
The set $\partial F(z):=\{x:\langle x,y-z\rangle\leq F(y)-F(z),\,\forall y\in Y\}$ is called \emph{subdifferential} of $F$ at $z$.  It follows immediately from the definition of subdifferential that $D_F[y,z]\geq0$.  We shall define $D_F[y,z]:=\infty$, if $\partial F(z)=\varnothing$ or $F(y)=\infty$.  We note that the notion of subdifferential can be applied to a non-convex function $F$.  However, non-empty $\partial F(z)$ implies $F(z)<\infty$ and $F(z)=F^{\ast\ast}(z)$, $\partial F(z)=\partial F^{\ast\ast}(z)$ (\cite{Rockafellar74}, Theorem~12).  Generally, $F^{\ast\ast}\leq F$, so that $F(y)-F(z)\geq F^{\ast\ast}(y)-F^{\ast\ast}(z)$ if $\partial F(z)\neq\varnothing$.  If $F$ is G\^{a}teaux differentiable at $z$, then $\partial F(z)$ has a single element $x=\nabla F(z)$, called the {\em gradient} of $F$ at $z$.  Thus, definition~(\ref{eq:distance}) is a generalization of the Bregman divergence for the case of a non-convex and non-differentiable $F$.  Note that the dual functional $F^\ast$ defines dual distance $D_F^\ast$ on $X$, which is related to $D_F$ as follows: $D_F[y,z]=D_F^\ast[\nabla F(z),\nabla F(y)]$.

\begin{theorem}
$D_F[y,z]=0\ \iff\ \{y,z\}\subseteq\partial F^\ast(x)\,,\ \exists\,x\in X$.
\label{th:zero-distance}
\end{theorem}

\begin{proof}
If $y=z$, then $D_F[y,z]=0$ by definition.  If $y\neq z$, then $\{y,z\}\subseteq\partial F^\ast(x)\,\iff\,\partial F(y)=\partial F(z)=\{x\}$, which follows from the property of subdifferentials: $y\in\partial F^\ast(x)\,\iff\,\partial F(y)\ni x$ (\cite{Rockafellar74}, Corollary to Theorem~12).  Thus, $D_F[y,z]=D_{F^\ast}[\nabla F(z),\nabla F(y)]=D_F^\ast[x,x]=0$.\qed
\end{proof}

\begin{corollary}
$D_F$ separates points of $\dom F\subseteq Y$ if and only if $F$ is G\^ateaux differentiable or $F^\ast$ is strictly convex.
\end{corollary}

Let us denote by $\nabla_1 D[y,z]$ and $\nabla_1^2 D[y,z]$ the first and the second G\^ateaux differentials of $D[y,z]$ with respect to the first argument.  For a twice G\^ateaux differentiable $F$ they are $\nabla_1 D_F[y,z]=\nabla F(y)-\nabla F(z)$ and $\nabla_1^2 D_F[y,z]=\nabla^2 F(y)$.

\begin{theorem}[Generalized Law of Cosines]
The following statements are equivalent:
\begin{eqnarray*}
D[y,z]&=&\int_0^1(1-t)\Bigl\langle\nabla_1^2 D[z+t(y-z),y](y-z),y-z\Bigr\rangle\,dt\\
D[y,w]&=&D[y,z]+D[z,w]-\langle\nabla_1 D[z,w],z-y\rangle
\end{eqnarray*}
\label{th:cosines}
\end{theorem}

\begin{proof}
Consider the first order Taylor expansion of $D[\cdot,w]$ at $z$:
\[
D[y,w]=D[z,w]+\langle\nabla_1 D[z,w],y-z\rangle+R_1[z,y]
\]
where the remainder is $R_1[z,y]=\int_0^1(1-t)\langle\nabla_1^2 D[z+t(y-z),y](y-z),y-z\rangle\,dt$.  The result follows from the equality $D[y,z]=R_1[z,y]$.\qed
\end{proof}

An asymmetric seminorm on space $X$ can be defined either by the gauge or support function of sublevel sets of distances $D_F^\ast[x,0]$ and $D_F[y,z]$ respectively:
\[
\|x|_{F^\ast}:=\inf\{\alpha>0:D_F^\ast[x/\alpha,0]\leq1\}\,,\quad
\|x|_F:=\sup_y\{\langle x,y-z\rangle:D_F[y,z]\leq1\}
\]
The supremum is achieved at $y(\beta)\in\partial F^\ast(\beta x)$, $D_F[y(\beta),z]=1$.  A quasi-pseudometric is defined as $\rho_{F^\ast}(w,x)=\|x-w|_{F^\ast}$ or $\rho_F(w,x)=\|x-w|_F$.  The dual space $Y$ is equipped with asymmetric seminorms and quasi-pseudometrics in the same manner.  The following characterization of the topology is known.
\begin{theorem}[\cite{Garcia-Raffi_etal03} or see Proposition~1.1.40 in \cite{Cobzas13}]
An asymmetric seminormed space $X$ is:
\begin{list}{}{}
\item [$T_0$] if and only if $\|x|_{F^\ast}>0$ or $\|-x|_{F^\ast}>0$ for all $x\neq0$;
\item [$T_1$] if and only if $\|x|_{F^\ast}>0$ for all $x\neq 0$;
\item [$T_2$ (Hausdorff)] if and only if $\|x\|_{F^\ast}^\circ>0$ for all $x\neq0$.
\end{list}
\end{theorem}

These separation properties depend on sublevel set $\{x:D_F^\ast[x,0]\leq1\}$.  For $T_0$ it must not contain any hyperplane; for $T_1$ it must not contain any ray (i.e. it must be bounded); for $T_2$ its polar set must contain zero in the interior (i.e. its polar must be absorbing).  The following theorem is useful in our analysis.

\begin{theorem}[\cite{Asplund-Rockafellar69,Moreau67}]
If $0\in\Int{\dom F^\ast}\subset X$, then sublevel sets $\{y:F(y)\leq\lambda\}$ are bounded.  Conversely, if one of the sublevel sets for $\lambda>\inf F$ is bounded, then $0\in\Int{\dom F^\ast}$.
\label{th:zero-interior-bounded}
\end{theorem}

\section{Asymmetric Topology Generated by the KL-Divergence}
\label{sec:KL-topology}

The KL-divergence can be defined as Bregman divergence associated with closed convex functional $KL(y)=\langle\ln y-1,y\rangle$:
\[
D_{KL}[y,z]=\langle\ln y-\ln z,y\rangle-\langle1,y-z\rangle
\]
Note that $KL$ is a proper closed convex functional that is finite for all $y\geq0$, if we define $(\ln 0)\cdot 0=0$ and $KL(y)=\infty$ for $y\ngeq0$. The dual of $KL$ is the moment generating functional $KL^\ast(x)=\langle e^x,z\rangle$.  The dual divergence of $x$ from $0$ is:
\[
D_{KL}^\ast[x,0]=\langle e^x-1-x,z\rangle
\]
The above divergence can be written as $D_{KL}^\ast[x,0]=\langle\phi^\ast(x),z\rangle$, where $\phi^\ast(x)=e^x-1-x$.  The dual of $\phi^\ast$ is the closed convex function $\phi(u)=(1+u)\ln(1+u)-u$.  Making the change of variables $y\mapsto u=\frac{y}{z}-1$, the KL-divergence can be written in terms of $\phi(u)$:
\[
D_{KL}[y,z]=D_{KL}[(1+u)z,z]=\langle(1+u)\ln(1+u)-u,z\rangle
\]
Sublevel set $M=\{y-z:D_{KL}[y,z]\leq1\}$ is a closed neighbourhood of $0\in Y-z$; sublevel set $N=\{x:D_{KL}^\ast[x,0]\leq1\}$ is a closed neighbourhoods of $0\in X$.  Both functions $\phi(u)$ and $\phi^\ast(x)$ are not even, and these neighbourhoods are not balanced.  In the theory of Orlicz spaces the symmetrized functions $\phi(|u|)$ and $\phi^\ast(|x|)$ are used to define even functionals and norms \cite{Krasnoselskii-Rutitskii58}.  This approach has been used in infinite-dimensional information geometry \cite{Pistone-Sempi95}.  In particular, because $\phi(|u|)$ belongs to the $\Delta_2$ class \cite{Krasnoselskii-Rutitskii58}, the corresponding Orlicz space $Y_{\phi(|\cdot|)}$ (and the statistical manifold it contains) is separable.  The dual Orlicz space $X_{\phi^\ast(|\cdot|)}$ is not separable, because $\phi^\ast(|x|)$ is not $\Delta_2$.  Note, however, that another symmetrization is possible: $\phi(-|u|)$, which is not $\Delta_2$, and $\phi^\ast(-|x|)$, which is $\Delta_2$.  Thus, one can introduce the non-separable Orlicz space $Y_{\phi(-|\cdot|)}$, and the dual separable Orlicz space $X_{\phi^\ast(-|\cdot|)}$.  One can check that the following inequalities hold: $\phi(|u|)\leq\phi(u)\leq\phi(-|u|)$ (resp. $\phi^\ast(|x|)\geq\phi^\ast(x)\geq\phi^\ast(-|x|)$), which corresponds to the following symmetrizations and inclusions of sublevel sets: $\co[-M\cup M]\supseteq M\supseteq -M\cap M$ (resp. $-N\cap N\subseteq N\subseteq\co[-N\cup N]$).  Thus, the asymmetric topology of space $Y_{\phi}$, induced by $D_{KL}$ (resp. of $X_{\phi^\ast}$, induced by $D_{KL}^\ast$) is finer than topology of the separable Orlicz space $Y_{\phi(|\cdot|)}$ (resp. $X_{\phi^\ast(-|\cdot|)}$), and so it is Hausdorff.  On the other hand, the diameter $\mathrm{diam}(M)=\sup\{\rho_{KL}(y,z):y,z\in M\}$ of set $M\subset Y$ (resp. for $\rho_{KL}^\ast(x,w)$ and $N\subset X$) is the diameter with respect to the metric of the Orlicz space $Y_{\phi(-|\cdot|)}$ (resp. $X_{\phi^\ast(|\cdot|)}$), which is complete.  Therefore, every nested sequence of sets with diameters decreasing to zero has non-empty intersection, so that space $Y_\phi$ (resp. $X_{\phi^\ast}$) is $\rho$-sequentially complete (\cite{Reilly1982}, Theorem~10).  Thus, we have proven the following theorem, concluding this short paper.

\begin{theorem}
Asymmetric seminorm $\|x|_{KL}^\ast:=\sup\{\langle x,y-z\rangle:D_{KL}[y,z]\leq1\}$ (resp. $\|y-z|_{KL}:=\inf\{\alpha^{-1}>0:D_{KL}[z+\alpha(y-z),z]\leq1\}$) induces Hausdorff topology on space $X$ (resp. on $Y=Y_+-Y_+$), and therefore it is an asymmetric norm.  It is $\rho$-sequentially complete and contains a separable subspace, which is an Orlicz space with the norm $\|x\|_{\phi^\ast(-|\cdot|)}$ (resp. $\|y-z\|_{\phi(|\cdot|)}$).
\label{th:kl-topology}
\end{theorem}

\bibliographystyle{splncs}

\bibliography{rvb,nn,other,newbib,ica,vpb}

\end{document}